\documentclass[draftclsnofoot,onecolumn,12pt]{IEEEtran}
\usepackage[T1]{fontenc}
\usepackage{cite}
\ifCLASSINFOpdf
   \usepackage[pdftex]{graphicx}
   \DeclareGraphicsExtensions{.pdf,.jpeg,.png}
   \usepackage{epstopdf}
\else
   \usepackage[dvips]{graphicx}
   \DeclareGraphicsExtensions{.eps}
\fi
\usepackage{amsmath}
\usepackage{amsthm}
\usepackage{amsfonts}
\usepackage{mathrsfs}
\usepackage{amssymb}
\usepackage{mdwmath}
\usepackage{multirow}
\usepackage{bm}
\theoremstyle{plain}

\newtheorem{lemma}{\textbf{Lemma}}
\newtheorem{theorem}{\textbf{Theorem}}
\newtheorem{proposition}{\textbf{Proposition}}
\newtheorem{corollary}{\textbf{Corollary}}
\theoremstyle{definition}
\newtheorem{definition}{\textbf{Definition}}
\theoremstyle{remark}
\newtheorem{example}{\textbf{Example}}
\newtheorem*{remark}{Remark}

\usepackage[algo2e,ruled,vlined]{algorithm2e}
\usepackage{mathtools}
\usepackage{array}

\ifCLASSOPTIONonecolumn
\newcommand{\figwidth}{0.65\textwidth}

\fi

\ifCLASSOPTIONtwocolumn
\newcommand{\figwidth}{0.48\textwidth}

\fi
\begin{document}
\title{Free Ride on LDPC Coded Transmission}
\author{Suihua~Cai,
        Shancheng~Zhao,
        and~Xiao~Ma,~\IEEEmembership{Member,~IEEE}
\thanks{This work was supported by the Science and Technology Planning Project of Guangdong Province (2018B010114001), the Basic Research Project of Guangdong Provincial NSF~(No.~2016A030308008), and the NSF of China~(No.~61871201 and No.~61771499).}
\thanks{S. Cai is with the School of Electronics and Information Technology, Sun Yat-sen University, Guangzhou 510006, China (e-mail: caish5@mail2.sysu.edu.cn).}
\thanks{S. Zhao is with the College of Information Science and Technology, Jinan University, Guangzhou 510632, China, and also with the Guangdong Key Laboratory of Data Security and Privacy Preserving and the Guangzhou Key Laboratory of Data Security and Privacy Preserving, Guangzhou 510632, China (e-mail: shanchengzhao@jnu.edu.cn).}
\thanks{X. Ma is with the School of Data and Computer Science and the Guangdong Key Laboratory of Information Security Technology, Sun Yat-sen University, Guangzhou 510275, China (corresponding author, e-mail: maxiao@mail.sysu.edu.cn).}
}
\maketitle

\begin{abstract}
In this paper, we formulate a new problem to cope with the transmission of extra bits over an existing coded transmission link~(referred to as coded payload link) without any cost of extra transmission energy or extra bandwidth.
This is possible since a gap to the channel capacity typically exists for a practical code.
A new concept, termed as {\em accessible capacity}, is introduced to specify the maximum rate at which the superposition transmission of extra bits is reliable and has a negligible effect on the performance of the coded payload link.
For a binary-input output-symmetric~(BIOS) memoryless channel, the accessible capacity can be characterized as the difference between the channel capacity and the mutual information rate of the coded payload link, which can be numerically evaluated for very short payload codes.
For a general payload code, we present a simple lower bound on the accessible capacity, given by the channel capacity minus the coding rate of the payload code.
We then focus on the scenarios where low-density parity-check~(LDPC) codes are implemented for the payload link.
We propose to transmit extra bits by random superposition for encoding, and exhaustive search~(with the aid of statistical learning) for decoding.
We further propose, by establishing an auxiliary channel~(called \emph{syndrome channel}) induced from ``zero-forcing'' over the binary field, to transmit extra bits with structured codes such as repetition codes and first-order Reed-Muller~(RM) codes.
Numerical results show that up to 60 extra bits can be reliably transmitted along with a rate-1/2 LDPC code of length 8064.
\end{abstract}
\IEEEpeerreviewmaketitle

\section{Introduction}
Low-density parity-check~(LDPC) codes, discovered by Gallager in 1962~\cite{Gallager1962}, are shown in practice to achieve near-capacity performance.
As a result, LDPC codes have been widely adopted in modern communication systems, such as the satellite transmission~\cite{DVBS2}, Wi-Fi~\cite{IEEE80211} and enhanced mobile broadband~(eMBB) applications in 5G~\cite{TS38212}.
In many scenarios, in addition to payload data, a few extra bits are required to be transmitted.
For example, in Digital Terrestrial Television~(DTV) systems, the transmitter identification~(TxID) is used to detect, diagnose and classify the operating status of radio transmitters.
Generally, it is required to acquire the TxID without decoding of the DTV signals, especially for weak interference sources~\cite{Wang2004}. In systems with hybrid automatic retransmission request~(HARQ), the ACK/NACK has to be transmitted to indicate correct reception or erroneous reception.
Similarly, the control signals can be viewed as extra data to be transmitted, as they typically have high reliability requirements.

One simple approach is to encode the extra bits and the payload data separately.
This separation is either because the receiver is only interested in the extra bits or because the reliability requirement of the extra bits is higher.
The main drawback of such an approach is that it will inevitably lead to an extra consumption of transmission power and bandwidth.
An interest question arises.
Is it possible to transmit extra bits with neither bandwidth expansion nor transmission energy increase?
The answer is theoretically positive since a gap always exists between the coding rate of the practical payload codes and the channel capacity.

For DTV system, digital watermarking techniques were used to embed TxID signals into the 8-symbol vestigial sideband~(8-VSB) signals~\cite{Wang2004}. This approach does not require extra bandwidth but slightly increases the transmission power.
In~\cite{Wang2005}, the authors proposed an additional data transmission scheme using TxID sequences.
In~\cite{Larsson2012}, a transmission scheme was presented to piggyback an additional lonely bit on coded payload data by using different interleavers.
In~\cite{Yan2011,Hong2015}, the authors proposed to use different constellations to transmit extra bits.
In~\cite{Xia2014}, blind identification of LDPC codes was proposed to avoid parameter transmission in adaptive coding and modulation.
For grant-free random access network, the authors in~\cite{Senel2018} proposed to embed short messages in pilot transmission.
To the best of our knowledge, existing schemes either require extra physical resources such as bandwidth and transmission power or can embed only a very small amount of extra bits into the payload transmission.

Recently, we have proposed in~\cite{Cai2019} a new coding scheme to transmit extra bits along with LDPC coded data, whereby extra bits are randomly encoded and then superimposed on LDPC-coded payload data.
The decoding of extra bits is achieved by exhaustive search with the aid of learning the statistics of the syndromes.
The simulation results in~\cite{Cai2019} show that, with a hard decision rule, ten extra bits can be embedded reliably into a rate-1/2 LDPC code with length $8064$.
This paper is more than an extension of~\cite{Cai2019}, having the following contributions.
\begin{enumerate}
\item We derive the accessible capacity of the superposition transmission of extra bits as well as a simple lower bound on the accessible capacity.
\item For random constructions, we propose an analytic expression to estimate the performance of the hard-decision decoding~(HDD) and further develop a soft-decision decoding~(SDD) to improve the performance.
\item We derive an auxiliary channel model~(termed as \emph{syndrome channel}) from a linear transformation defined at the receiver with the parity-check matrix of the LDPC code.
    We then propose a systematic approach to transmission of extra bits by using structured codes.
\end{enumerate}

To verify the basic ideas, we take repetition codes and first-order Reed-Muller~(RM) codes as construction examples.
Simulation results show that as many as 60 extra bits can be packed into a rate-1/2 LDPC code with length 8064, while having a negligible effect on the reliability of the payload data.
The proposed extra data transmission schemes can be used in practical systems to efficiently transmit extra bits.

The rest of this paper is organized as follows.
The problem statements and the system model are introduced in Section II.
In the presented model, the accessible capacity for the extra transmission is derived, where also presented is a constructive proof for the lower bound on the accessible capacity.
The coding schemes for extra transmission based on random superposition are presented in Section III.
In Section IV, new coding schemes using structured codes are proposed for extra transmission, where both encoding and decoding are designed for the derived syndrome channel.
Finally, Section V summarizes this paper.
\section{Accessible Capacity}
\subsection{Payload Transmission}
We focus on a coded payload link consisting of a binary linear block code and a binary-input output-symmetric~(BIOS) memoryless channel~\cite{Viterbi2013}.
The detailed formulation is presented as follows.
\begin{enumerate}
    \item[1)] \emph{Encoding}\\
        Let $\mathbb{F}_2$ be the binary field and $\mathscr{C}_0[n,k]$ be a binary linear block code with dimension $k$ and length $n$, referred to as the \emph{payload code}.
        Let $\bm{u}=(\bm{u}^{(0)},\bm{u}^{(1)},\dots,\bm{u}^{(L-1)}) \in\mathbb{F}_2^{kL}$ be the payload data to be transmitted, where $\bm{u}^{(t)}\in\mathbb{F}_2^k$ are data blocks with size $k$.
        By performing separately $L$ times the encoding algorithm of $\mathscr{C}_0[n,k]$, the payload data are transformed to a sequence of codewords $\bm{c}=(\bm{c}^{(0)},\bm{c}^{(1)},\dots,\bm{c}^{(L-1)}) \in\mathbb{F}_2^{nL}$, where $\bm{c}^{(t)}\in \mathscr{C}_0[n,k]$ is the codeword corresponding to the data block $\bm{u}^{(t)}$.
    \item[2)] \emph{BIOS Memoryless Channel}\\
        A BIOS channel is characterized by an input set $\mathbb{F}_2$, an output set $\mathcal{Y}$~(discrete or continuous), and a conditional probability mass~(or density) function $P_{Y|X}(y|x),x\in\mathbb{F}_2,y\in\mathcal{Y}$, satisfying the symmetric condition that
        \begin{equation}
            P_{Y|X}(y|1)=P_{Y|X}(\pi(y)|0)
        \end{equation}
        for some mapping $\pi:\mathcal{Y}\rightarrow\mathcal{Y}$ with $\pi^2(y)=y$ for all $y\in\mathcal{Y}$.
        The transmitted sequence $\bm{c}$ is transmitted over a BIOS memoryless channel, resulting in a received sequence $\bm{y}\in \mathcal{Y}^{nL}$.
    \item[3)] \emph{Decoding}\\
        Upon receiving $\bm{y}$, the receiver attempts to estimate the transmitted payload bits $\bm{u}$.
        Since the channel is memoryless, this can be done separately by performing $L$ times the decoding algorithm of the payload code.
        For notational convenience, we denote by $\psi_0$ the decoder which accepts $\bm{y}$ as input and delivers as output the estimated payload sequence $\widehat{\bm{u}}=(\widehat{\bm{u}}^{(0)},\widehat{\bm{u}}^{(1)},\dots,\widehat{\bm{u}}^{(L-1)})$, where $\widehat{\bm{u}}^{(t)} \in \mathbb{F}_2^k$ for $t=0,1,\dots,L-1$.
    \item[4)] \emph{Error Performance Metric}\\
        We use the word error rate~(WER) and/or bit error rate~(BER) to characterize the performance of the (de)coding scheme.
        We define random variables
        \begin{equation}
            E^{(t)}_i=\left\{ \begin{array}{ll}
                0, & \textrm{if}~\widehat{U}_i^{(t)}=U_i^{(t)}\\
                1, & \textrm{if}~\widehat{U}_i^{(t)}\neq U_i^{(t)}
            \end{array} \right.  ~~\textrm{for}~0\leqslant i\leqslant k-1~\textrm{and}~0\leqslant t \leqslant L-1 .
        \end{equation}
        Then the WER and BER are given by
        \begin{equation}\label{eqn:WERpl}
            \mathrm{WER}=\frac{1}{L}\sum_{t=0}^{L-1}\textrm{Pr}\{\widehat{\bm{U}}^{(t)} \neq \bm{U}^{(t)} \}=\frac{1}{L}\textbf{E}[\sum_{t=0}^{L-1}\max_{0\leqslant i\leqslant k-1}E_i^{(t)}]
        \end{equation}
        and
        \begin{equation}\label{eqn:BERpl}
            \mathrm{BER}=\frac{1}{kL}\sum_{t=0}^{L-1}\sum_{i=0}^{k-1}\textrm{Pr}\{\widehat{U}_i^{(t)} \neq {U}_i^{(t)} \}=\frac{1}{kL}\textbf{E}[\sum_{t=0}^{L-1}\sum_{i=0}^{k-1}E_i^{(t)}],
        \end{equation}
        where $\textbf{E}[\cdot]$ is the mathematical expectation.
        Usually,  WER is taken as the performance metric for short block codes while BER is for long block codes.
\end{enumerate}

\subsection{Extra Transmission}
Let $\bm{v}\in \mathbb{F}_2^{k_1}$ be $k_1$ extra bits to be transmitted.
We are seeking a way to transmit $\bm{v}$ along with $\bm{u}$ without any bandwidth expansion or transmission power increase.
One way is to construct a \emph{totally new} encoding function $\phi$ that maps $(\bm{u}, \bm{v})$ into $\mathbb{F}_2^{nL}$.
We use the term ``totally new'' to suggest that such an encoding function can be irrelevant to $\mathscr{C}_0$.
This, however, need change the encoder for the payload transmission and might be inconvenient in some applications.
In this paper, we focus on implementing the transmission of extra bits by superposition, as shown in Fig.~\ref{fig:model} and described as follows.
\begin{figure}
    \centering
    \includegraphics[width=\figwidth]{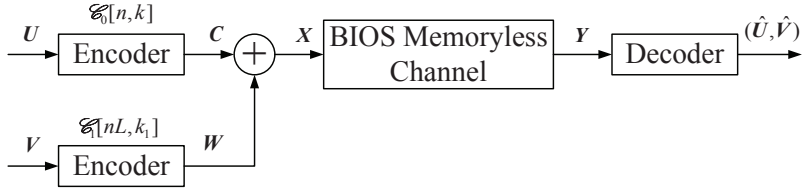}
    \caption{Diagram of system model.}
    \label{fig:model}
\end{figure}
\begin{enumerate}
    \item[1)] \emph{Encoding}\\
        The payload data sequence $\bm{u} \in\mathbb{F}_2^{kL}$ is encoded by the payload code $\mathscr{C}_0[n,k]$, while the extra bits $\bm{v}\in \mathbb{F}_2^{k_1}$ are encoded by a binary block code $\mathscr{C}_1[nL,k_1]$ with dimension $k_1$, resulting in a codeword $\bm{w}$ with length $nL$.
        The coded sequence $\bm{x}$ is given by
        \begin{equation}
            \phi(\bm{u},\bm{v})=\bm{x}=\bm{c}+\bm{w},
        \end{equation}
    where ``$+$'' is the component-wise addition over $\mathbb{F}_2$.
    The coding rates of the payload and extra bits are $R_0=k/n$ and $R_1=k_1/(nL)$, respectively.
    \item[2)] \emph{BIOS Memoryless Channel}\\
        The coded sequence $\bm{x}$ is transmitted over a BIOS memoryless channel, resulting in a received sequence $\bm{y}\in \mathcal{Y}^{nL}$.
    \item[3)] \emph{Decoding}\\
        Given $\bm{y}$, the receiver attempts to decode both the messages $\bm{u}$ and $\bm{v}$.
        This can be achieved naturally by implementing a successive cancellation decoder.
        First, a decoder $\psi_1$ is used to decode $\bm{v}$ from $\bm{y}$, i.e.,
        \begin{equation}\label{eqn:SC0}
            \widehat{\bm{v}}=\psi_1(\bm{y}).
        \end{equation}
        Then, by removing the effect of extra coded data, $\bm{u}$ is recovered by the decoder $\psi_0$ of $\mathscr{C}_0$.
        That is,
        \begin{equation}
            \widehat{\bm{u}}=\psi_0(\bm{y}\boxplus\widehat{\bm{w}}),
        \end{equation}
        where $\widehat{\bm{w}}$ is the $\mathscr{C}_2[nL,k_1]$ codeword corresponding to $\widehat{\bm{v}}$ and the operation $\boxplus$, referred to as \emph{flipping operation}, is defined as
        \begin{equation}
            y_i\boxplus w_i=\left\{
            \begin{array}{ll}
                y_i,&w_i=0\\
                \pi(y_i),&w_i=1
            \end{array}
            \right. ~~\textrm{for}~0\leqslant i \leqslant nL-1.
        \end{equation}
        For integrity, we denote the decoding function by
        \begin{equation}\label{eqn:SC}
        \psi(\bm{y})=(\psi_0(\bm{y}\boxplus\widehat{\bm{w}}),\psi_1(\bm{y})).
        \end{equation}
    \item[4)] \emph{Error Performance Metric}\\
        For technical reasons concerning the proofs below, we take WER $\lambda_1^{(L)}=\mathrm{Pr}\{\widehat{\bm{V}}\neq \bm{V}\}$ of extra bits as performance metric for the superposition transmission.
\end{enumerate}

We refer to the code $\mathscr{C}_1$ for the extra transmission as the \emph{free-ride code}.
By ``free-ride'', we mean that the superposition transmission of extra bits requires neither extra bandwidth nor extra transmission power.
We will further prove that the extra transmission with properly designed free-ride codes can have a negligible effect on the error performance of the payload transmission.
Certainly, nothing is free since we need extra computational loads for encoding/decoding the extra bits.

\begin{definition}
    A rate $R_1$ is \emph{accessible} for the extra transmission, if there exists a sequence of coding/decoding functions  $(\mathscr{C}_1,\psi_1)$ with coding rates no less than $R_1$, such that $\lim_{L\rightarrow\infty}\lambda_1^{(L)}=0$.
\end{definition}
\begin{definition}
    The \emph{accessible capacity} for the extra transmission is defined as the supremum of all accessible rates, i.e.,
    \begin{equation}
        C_a(\mathscr{C}_0)=\sup\{R_1:R_1~\textrm{is~accessable}\}.
    \end{equation}
\end{definition}
The notation $C_a(\mathscr{C}_0)$ indicates that the accessible capacity depends on the payload code $\mathscr{C}_0$.
When no confusion arises in the context, we use $C_a$ instead for ease.

\begin{remark}
The accessible capacity defined in this paper is similar to but different from that defined in~\cite{Huang2014}.
In~\cite{Huang2014}, the Gaussian interference channel is considered, where the superposition of the signals occurs at the receiver and is over the real field.
As a result, the transmission of the secondary user in~\cite{Huang2014} requires extra transmission power.
In contrast, the transmission of extra bits is considered in this paper and the superposition is tailored over the binary field, which requires no extra transmission power.
\end{remark}

We will show by examples that $C_a$ can be calculated numerically for short block codes $\mathscr{C}_0$ and prove that the extra transmission at an accessible rate can have a negligible effect on the payload transmission in the following sense.
\begin{proposition}
    Let $\lambda_0$ and $\widetilde{\lambda}_0$ be the WERs with respect to the payload transmission for the scenarios with and without extra transmission, respectively.
    Let $\lambda_1^{(L)}$ be the WER with respect to the extra transmission.
    Then $\lambda_0$ can be bounded by
    \begin{equation}
        \lambda_0 \leqslant \widetilde{\lambda}_0+\lambda_1^{(L)},
    \end{equation}
    indicating that the effect of the extra transmission can be negligibly small when the extra transmission is reliable.
    That is, this upper bound implies that $\lambda_0\rightarrow \widetilde{\lambda}_0$ as $\lambda_1^{(L)}\rightarrow 0$ since $\lambda_0 \geqslant \widetilde{\lambda}_0$ holds obviously.
\end{proposition}
\begin{proof}
    Consider the successive cancellation decoder as described in~(\ref{eqn:SC0})-(\ref{eqn:SC}).
    Let $\widehat{\bm{W}}$ be the estimated free-ride codeword corresponding to $\widehat{\bm{V}}=\psi_1(\bm{Y})$. Then we have
    \begin{align}
        \notag \lambda_0 =&\frac{1}{L}\sum_{t=0}^{L-1}\mathrm{Pr}\{\psi_0^{(t)}(\bm{Y}\boxplus\widehat{\bm{W}})\neq \bm{U}^{(t)}\}\\
        \notag\leqslant& \frac{1}{L}\sum_{t=0}^{L-1}\mathrm{Pr}\{\psi_0^{(t)}(\bm{Y}\boxplus\widehat{\bm{W}})\neq \bm{U}^{(t)},\psi_1(\bm{Y})=\bm{V}\}+\mathrm{Pr}\{\psi_1(\bm{Y})\neq\bm{V}\}\\
        \notag\leqslant& \frac{1}{L}\sum_{t=0}^{L-1}\mathrm{Pr}\{\psi_0^{(t)}(\bm{Y}\boxplus\bm{W})\neq \bm{U}^{(t)}\}+\mathrm{Pr}\{\widehat{\bm{V}}\neq\bm{V}\}\\
        =&\widetilde{\lambda}_0+\lambda_1^{(L)},
    \end{align}
    where $\psi_0^{(t)}(\bm{Y}\boxplus\widehat{\bm{W}})$ is the $t$-th sub-block of $\psi_0(\bm{Y}\boxplus\widehat{\bm{W}})$, i.e., $\psi_0^{(t)}(\bm{Y}\boxplus\widehat{\bm{W}})=\widehat{\bm{U}}^{(t)}$.
\end{proof}
For ease of analysis, we assume that the payload data $\{U_i^{(t)}:0\leqslant i<k,0\leqslant t<L\}$ are independent and uniformly distributed~(i.u.d.) binary random variables.
Then the accessible capacity can be calculated from the maximum mutual information over the transmission link $\bm{w} \rightarrow \bm{y}$.
\begin{theorem}\label{thm:1}
    Let $\mathscr{C}_0$ be the payload code with coding rate $R_0=k/n$ over the given BIOS memoryless channel with capacity $C_{\textrm{BIOS}}$.
    The accessible capacity $C_a$ for the extra transmission is given by
    \begin{equation}
        C_a=C_{\textrm{BIOS}}-\frac{1}{n}I(\bm{C}^n;\widetilde{\bm{Y}}^n),
    \end{equation}
    where $I(\bm{C}^n;\widetilde{\bm{Y}}^n)$ is the mutual information between the transmitted codeword $\bm{C}^n \in\mathscr{C}_0$ and the received sequence $\widetilde{\bm{Y}}^n$ over the given BIOS channel without extra transmission.
\end{theorem}
\begin{proof}
    Referring to Fig.~\ref{fig:model}, we consider the transmission link from the coded extra data $\bm{w}$ to the received sequence $\bm{y}$, where the coded payloads are treated as channel noise.
    This transmission link can be viewed as a block-wise stationary memoryless  channel with block length $n$, where the input set is $\mathbb{F}_2^{n}$, the output set is $\mathcal{Y}^n$, and the transition probability mass~(density) function is given by
    \begin{equation}\label{eqn:blockchannel}
        P_{\bm{Y}^n|\bm{W}^n}(\bm{y}|\bm{w})=2^{-k}\sum_{\bm{c}\in \mathscr{C}_0}P_{\bm{Y}^n|\bm{X}^n}(\bm{y}|\bm{w}+\bm{c}).
    \end{equation}
    Then the accessible capacity can be calculated as
    \begin{align}
        \notag C_a=&\frac{1}{n}\max_{p(\bm{w}^n)}I(\bm{W}^n;\bm{Y}^n)\\
        \notag =&\frac{1}{n}\max_{p(\bm{w}^n)} [H(\bm{Y}^n)-H(\bm{Y}^n|\bm{W}^n)]\\
        \notag =&\frac{1}{n}\max_{p(\bm{w}^n)} [H(\bm{Y}^n)-H(\bm{Y}^n|\bm{C}^n,\bm{W}^n)+H(\bm{Y}^n|\bm{C}^n,\bm{W}^n)-H(\bm{Y}^n|\bm{W}^n)]\\
        =&\frac{1}{n}\max_{p(\bm{w}^n)} [I(\bm{C}^n,\bm{W}^n;\bm{Y}^n)-I(\bm{C}^n;\bm{Y}^n|\bm{W}^n)]. 
    \end{align}

    It is not difficult to verify that, for BIOS, the conditional mutual information $I(\bm{C}^n;\bm{Y}^n|\bm{W}^n)$ does not depend on the distribution $p(\bm{w}^n)$ and is equal to  $I(\bm{C}^n;\widetilde{\bm{Y}}^n)$, where $\widetilde{\bm{Y}}^n$ denotes the received sequence without extra data superposition.
    Also notice that
    \begin{equation}
        I(\bm{C}^n,\bm{W}^n;\bm{Y}^n)= I(\bm{C}^n,\bm{W}^n,\bm{X}^n;\bm{Y}^n)= I(\bm{X}^n;\bm{Y}^n),
    \end{equation}
    since $(\bm{C}^n,\bm{W}^n)\rightarrow\bm{X}^n\rightarrow\bm{Y}^n$ forms a Markov chain.
    So we have
    \begin{align}
        \notag C_a =&\frac{1}{n}\left(\max_{p(\bm{w}^n)}I(\bm{X}^n;\bm{Y}^n)-I(\bm{C}^n;\widetilde{\bm{Y}}^n)\right).\\
        =& C_{\textrm{BIOS}}-\frac{1}{n}I(\bm{C}^n;\widetilde{\bm{Y}}^n),
    \end{align}
    where the last equality holds because of an important fact that, when $\bm{W}^n$ is uniformly distributed over $\mathbb{F}_2^n$, the coded sequence $\bm{X}^n=\bm{C}^n+\bm{W}^n$ is also uniformly distributed for whatever distribution of $\bm{C}^n$.
%
\end{proof}

\begin{figure}
    \centering
    \includegraphics[width=\figwidth]{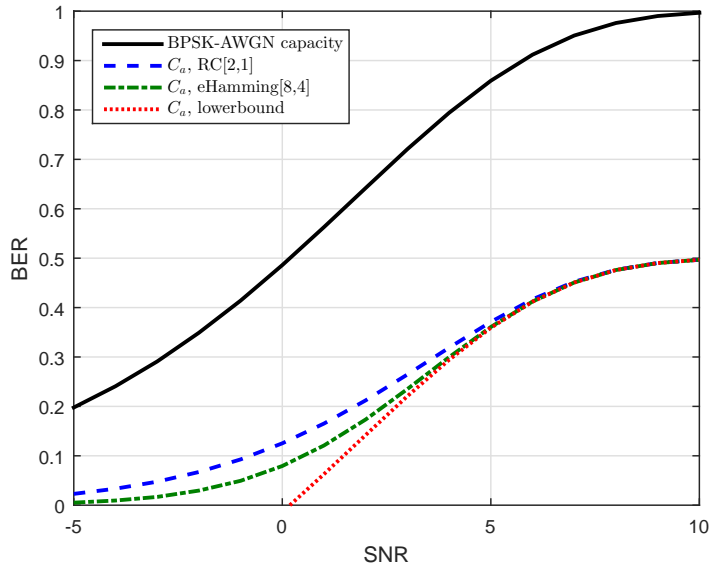}
    \caption{Accessible capacity of rate-1/2 codes over BPSK-AWGN channels.}
    \label{fig:ca}
\end{figure}

\begin{example}
    To illustrate the accessible capacity, we consider simple payload codes over a discrete time additive white Gaussian noise~(AWGN) channel with binary phase shift keying~(BPSK) signalling.
    Two codes with rate $1/2$, repetition code $\mathscr{C}[2,1]$ and the extended Hamming code $\mathscr{C}[8,4]$, are used for the simulation, as shown in Fig.~\ref{fig:ca}.
    We see that the accessible capacity increases as the signal-to-noise ratio~(SNR) increases.
    We can also see that the accessible capacity with respect to the repetition code is higher than that with respect to the extended Hamming code, suggesting that the weaker the payload code is, the higher the accessible capacity is.
    Furthermore, we see that, in high SNR region, the accessible capacities with respect to both codes are well predicted by the lower bound derived in the next subsection.
\end{example}

\subsection{A Lower Bound on Accessible Capacity}
From \textbf{Theorem~\ref{thm:1}}, we see that the accessible capacity is related to the mutual information $I(\bm{C}^n;\widetilde{\bm{Y}}^n)$, which is usually difficult to evaluate for long codes.
For this reason, we provide here a simple lower bound on the accessible capacity.

\begin{theorem}\label{thm:2}
    Let $\mathscr{C}_0$ be a block code with coding rate $R_0=k/n$ over the given BIOS memoryless channel with capacity $C_{\textrm{BIOS}}$.
    The accessible capacity $C_a$  is lower bounded by
    \begin{equation}
        C_a\geqslant C_{\textrm{BIOS}}-R_0.
    \end{equation}
\end{theorem}
\begin{proof}
    Since $I(\bm{C}^n;\widetilde{\bm{Y}}^n)\leqslant H(\bm{C}^n)=H(\bm{U}^k)\leqslant k$, we have, from \textbf{Theorem~\ref{thm:1}},
    \begin{align}
        C_a = C_{\textrm{BIOS}}-\frac{1}{n}I(\bm{C}^n;\widetilde{\bm{Y}}^n)\geqslant C_{\textrm{BIOS}}-\frac{k}{n}= C_{\textrm{BIOS}}-R_0.
    \end{align}
\end{proof}

The above proof is simple but implies that the lower bound can be tight for the payload codes with low WERs, in which case we can prove that $I(\bm{C}^n;\widetilde{\bm{Y}}^n) \approx k$ by invoking Fano's inequality.
Next, we present an alternative proof that is constructive and provides more insights on the construction of the free-ride codes.
The basic technique for the proof, similar to that employed in~\cite{Ma2016}, is to apply Markov's inequality to the decoding error probability and then to invoke the weak law of large numbers~(WLLN) as stated in the following lemma.

\begin{lemma}\label{lem:1}
    Suppose that the zero sequence $\bm{0}^{nL}$ is transmitted over the BIOS channel with capacity $C_{\textrm{BIOS}}$.
    As $L$ goes to infinity, $\frac{1}{nL}\log\frac{P_{\bm{Y}^{nL}|\bm{X}^{nL}}(\bm{Y}^{nL}|\bm{0}^{nL})}{P_{\bm{Y}^{nL}}(\bm{Y}^{nL})}$ converges to $C_{\textrm{BIOS}}$ in probability, where $P_{\bm{Y}^{nL}}(\bm{Y}^{nL})$ corresponds to the input distribution $P_X(1) = 1/2$.
    That is, for any given $\delta,\varepsilon>0$, with $L$ sufficiently large, we have
    \begin{equation}
        \mathrm{Pr}\left\{\left|\frac{1}{nL}\log\frac{P_{\bm{Y}^{nL}|\bm{X}^{nL}}(\bm{Y}^{nL}|\bm{0}^{nL})}{P_{\bm{Y}^{nL}}(\bm{Y}^{nL})}-C_{\textrm{BIOS}}\right|<\delta\right\}>1-\varepsilon.
    \end{equation}
\end{lemma}
\begin{proof}
    For a BIOS channel, the capacity $C_{\textrm{BIOS}}$ can be calculated by
    \begin{equation}
        C_{\textrm{BIOS}}=\textbf{E}\left[\log\frac{P_{Y|X}(Y|0)}{P_{Y}(Y)}\right],
    \end{equation}
    showing that this lemma is the application of the WLLN.
\end{proof}

\begin{proof}[Alternative Proof of Theorem \ref{thm:2}]
We need prove that any rate $R_1<C_{\textrm{BIOS}}-R_0$ is accessible.
For each $L\geqslant 1$, we construct a \emph{totally random} linear block code with dimension $k_1=\lceil nLR_1\rceil$ and length $nL$ as the free-ride code $\mathscr{C}_1$.
Let $\bm{G}_0$ and $\bm{G}_1$ be the generator matrices of $\mathscr{C}_0[n,k]$ and $\mathscr{C}_1[nL,k_1]$, respectively.
Then the payload data and the extra bits $(\bm{u},\bm{v})$ are encoded by a linear block code $\mathscr{C}[nL,kL+k_1]$ with the following generator matrix
\begin{equation}\label{eqn:semirandom}
    \bm{G}=\left[\begin{matrix}\widetilde{\bm{G}}_0\\ \bm{G}_1\end{matrix}\right],
\end{equation}
where $\widetilde{\bm{G}}_0$ is the generator matrix of the $L$-fold Cartesian product of $\mathscr{C}_0[n,k]$
\begin{equation}
    \widetilde{\bm{G}}_0=\left[\begin{array}{cccc}
               \bm{G}_0 &  &  &  \\
                & \bm{G}_0 &  &  \\
                &  & \ddots &  \\
                &  &  & \bm{G}_0
             \end{array}\right],
\end{equation}
and $\bm{G}_1$ is the generator matrix of the free-ride code with elements being generated independently according to the Bernoulli distribution with success probability $1/2$.

Without loss of generality, we assume that the all-zero codeword $\bm{0}^{nL}$ is transmitted.
We also assume that the output space $\mathcal{Y}$ is discrete for notational simplicity.
A maximum likelihood~(ML) decoding algorithm $\psi_1$ is assumed to decode the extra bits, which outputs $\bm{v}$ if, for some $\bm{u}$, the codeword $\phi(\bm{u},\bm{v})$ is the most likely one, i.e.,
\begin{equation}
    P_{\bm{Y}^{nL}|\bm{X}^{nL}}(\bm{y}^{nL}|\phi(\bm{u},\bm{v}))=\max_{\bm{x}^{nL}\in \mathscr{C}}P_{\bm{Y}^{nL}|\bm{X}^{nL}}(\bm{y}^{nL}|\bm{x}^{nL}).
\end{equation}

Define $\delta=(C_{\textrm{BIOS}}-R_0-R_1)/2$ and
\begin{equation}
\mathcal{A}^{(L)}=\left\{\bm{y}^{nL}: \left|\frac{1}{nL}\log\frac{P_{\bm{Y}^{nL}|\bm{X}^{nL}}(\bm{y}^{nL}|\bm{0}^{nL})}{P_{\bm{Y}^{nL}}(\bm{y}^{nL})}-C_{\textrm{BIOS}}\right|<\delta\right\}.
\end{equation}
From \textbf{Lemma~\ref{lem:1}}, for any $\varepsilon>0$, there exists an integer $L_0$ such that $\mathrm{Pr}\{\mathcal{A}^{(L)}\}>1-\varepsilon$ whenever $L>L_0$.
Then from the law of total expectation,
\begin{align}
\notag \lambda_1^{(L)}&=\sum_{\bm{y}^{nL}\in\mathcal{Y}^{nL}}\mathrm{Pr}\{\bm{y}^{nL}\}\cdot\lambda_1^{(L)}|_{\bm{y}^{nL}}\\
&<\varepsilon+\sum_{\bm{y}^{nL}\in \mathcal{A}^{(L)}}\mathrm{Pr}\{\bm{y}^{nL}\}\cdot\lambda_1^{(L)}|_{\bm{y}^{nL}}~~~\textrm{for}~L>L_0.
\label{eqn:lam1_1}
\end{align}
where $\lambda_1^{(L)}|_{\bm{y}^{nL}}$ is the WER of extra transmission conditional on the received sequence $\bm{y}^{nL}$.

Given the received vector $\bm{y}^{nL}$, the ML codeword $\bm{X}^{nL}$ is a random vector over the code ensemble due to the randomness of $\bm{G}_1$.
In particular, the ML codeword $\bm{X}^{nL}=\bm{u}\widetilde{\bm{G}}_0 + \bm{v}\bm{G}_1$ is an i.u.d. binary random variable sequence once the decoding output $\bm{v}$ is erroneous, i.e., $\bm{v}\neq \bm{0}^{k_1}$.
Hence, by using the same trick as in~\cite{Ma2016}, the conditional WER of extra transmission can be upper bounded by
\begin{align}
    \notag \lambda_1^{(L)}|_{\bm{y}^{nL}}=&\sum_{\bm{v}\neq \bm{0}^{k_1}}\mathrm{Pr}\{\bm{v}~\mathrm{is~the~decoding~output}|\bm{y}^{nL}\}\\
    \notag \leqslant &\sum_{(\bm{u},\bm{v}):\bm{v}\neq \bm{0}^{k_1}}\mathrm{Pr}\{P_{\bm{Y}^{nL}|\bm{X}^{nL}}(\bm{y}^{nL}|\bm{X}^{nL})\geqslant P_{\bm{Y}^{nL}|\bm{X}^{nL}}(\bm{y}^{nL}|\bm{0}^{nL})\}\\
    \notag\leqslant &2^{kL+k_1}\frac{\textbf{E}[P_{\bm{Y}^{nL}|\bm{X}^{nL}}(\bm{y}^{nL}|\bm{X}^{nL})]}{P_{\bm{Y}^{nL}|\bm{X}^{nL}}(\bm{y}^{nL}|\bm{0}^{nL})}~~~~\textrm{(by~Markov's~inequality)}\\
    \notag =&2^{kL+k_1} \frac{P_{\bm{Y}^{nL}}(\bm{y}^{nL})}{P_{\bm{Y}^{nL}|\bm{X}^{nL}}(\bm{y}^{nL}|\bm{0}^{nL})}\\
    \notag <&2^{kL+k_1-nL(C_{\textrm{BIOS}}-\delta)}~~~~\textrm{(by~\textbf{Lemma~1})}\\
    =& 2^{-nL(C_{\textrm{BIOS}}-R_0-\frac{k_1}{nL}-\delta)}.
     \label{eqn:lam1_2}
\end{align}
Since
\begin{equation}
    \lim_{L\rightarrow \infty}\left(C_{\textrm{BIOS}}-R_0-\frac{k_1}{nL}-\delta\right)= \delta >0,
\end{equation}
then from~(\ref{eqn:lam1_1}) and~(\ref{eqn:lam1_2}), we have $\lambda_1^{(L)} \rightarrow 0$ as $L\rightarrow \infty$.
\end{proof}

We borrow the term \emph{semi-random} from~\cite{Lin2018} to characterize a block linear code that has a generator matrix consisting of both random rows and non-random rows~(see~(\ref{eqn:semirandom}) for an example).
Then we have, as a by-product, the following corollary.
\begin{corollary}
    Let $\mathscr{C}$ be a semi-random block code with coding rate less than the BIOS channel capacity.
    Asymptotically, the bits corresponding to the random rows can be decoded arbitrarily reliable.
\end{corollary}

\begin{proof}
    The proof is omitted here.
\end{proof}

\section{Random Free-ride Codes}
This section serves to present the original idea of this work by reviewing our previous work.
As new achievements, we propose an analytical expression for estimating the HDD performance and develop the SDD for improving the performance.

Theoretically, the extra transmission at rates below the accessible capacity, requiring neither bandwidth expansion nor transmission power increase, can have a negligible effect on the payload transmission.
However, to recover one sub-block of the payload data, one must receive all the $L$ involved noisy sub-blocks and decode the extra bits first.
Apparently, this extra delay is not desirable in practice for large $L$.
In the remainder of this paper, we focus on $L=1$ and assume an LDPC code for the payload link.
We also assume that the coded sequence $\bm{x}$ is modulated by BPSK signalling and transmitted over an AWGN channel, resulting in a receiving vector
\begin{equation}
    \bm{y}=(-1)^{\bm{x}}+\bm{z},
\end{equation}
where $(-1)^{\bm{x}}$ is the transmitted signal sequence with the $i$-th component being $(-1)^{x_i}$ and $\bm{z}$ is a noise vector with each component being drawn independently from a normal distribution with mean zero and variance $\sigma^2$.
\subsection{Syndrome Statistics}
In our previous work~\cite{Cai2019}, we have presented a coding scheme based on random superposition to pack extra bits into LDPC coded data.
Similar to~\cite{Lin2019}, the decoding of the superposition transmission is aided by an off-line statistical learning.
The basic idea is to distinguish the case without  from that with superposition of extra bits by the statistical behavior of the syndrome.
The detailed analysis is shown as follows.

Let $\bm{H}=(h_{ij})_{m\times n}$ be the parity-check matrix of the LDPC code for the payload transmission, where $m = n-k$.
For ease of analysis, we assume that the LDPC code is $(\gamma,\rho)$-regular, i.e., $\bm{H}$ has constant column weight $\gamma$ and constant row weight $\rho$.
The payload data $\bm{u}$ are encoded by the LDPC code, while the extra bits $\bm{v}$ are encoded by a random linear block code $\mathscr{C}_1$ of length $n$.
Then the transmitted codeword is given by $\bm{x} = \bm{u}\bm{G}_0+\bm{v}\bm{G}_1$, where $\bm{G}_0$ is the generator matrix of the payload LDPC code and $\bm{G}_1$ is the generator matrix of the free-ride code with elements being generated independently according to the Bernoulli distribution with success probability $1/2$.
At the receiver, we make hard-decisions on the received signal $\bm{y}$, resulting in $\bm{\widehat{y}}$.
Then we define $N(\bm{s})\triangleq W_H((\bm{\widehat{y}}+\bm{s})\bm{H}^{\mathrm{T}})$ for each $\bm{s}\in\mathscr{C}_1$, where $\bm{H}^{\mathrm{T}}$ denotes the transpose of  $\bm{H}$ and $W_H(\cdot)$ denotes the Hamming weight function.
In words, $N(\bm{s})$ is the number of \emph{unsatisfied} parity checks if $\widetilde{\bm{c}} = \bm{\widehat{y}}+\bm{s}$ is viewed as a noisy version of the transmitted LDPC codeword.
Due to the existence of the channel noise and the considered random-like free-ride code, $N(\bm{s})$ can be treated as a random variable, whose distribution heavily depends on whether  the assumed free-ride codeword $\bm{s}$ is the transmitted one or not.
To be precise, we distinguish the following two cases.

\textbf{Case 1.} The free-ride codeword $\bm{s}$ is the transmitted one, i.e., $\bm{s}=\bm{w}=\bm{v}\bm{G}_1$.
In this case, we have
\begin{equation}
    \widetilde{\bm{c}}=\bm{\widehat{y}}+\bm{s}=\bm{x}+\bm{\widehat{z}}+\bm{s}=\bm{u}\bm{G}_0+\bm{\widehat{z}},
\end{equation}
where $\bm{\widehat{z}}$ is the error pattern corresponding to the hard decision.
Hence $\widetilde{\bm{c}}$ can be viewed as the received sequence by passing the LDPC codeword $\bm{c}$ through a binary symmetric channel~(BSC) with
crossover probability
\begin{equation}
    p_b=Q\left(\frac{1}{\sigma}\right),
\end{equation}
where $Q(x)=\int_x^{+\infty}\frac{1}{\sqrt{2\pi}}\exp(-\frac{1}{2}t^2)\mathrm{d}t$.
The probability that a parity check is not satisfied can be calculated as
\begin{align}\label{eqn:calpc}
    \notag p=&\sum_{j~\textrm{is~odd}}\binom{\rho}{j}p_b^j(1-p_b)^{\rho-j}\\
    =&\frac{1}{2}[1-(1-2p_b)^{\rho}].
\end{align}
Then the distribution of $N_{\bm{s}}$ can be approximated as\footnote{This approximation would be accurate only when $\gamma = 1$,  i.e., none of the parity-check equations were ``overlapped''. }
\begin{equation}\label{eqn:Nw}
    \mathrm{Pr}\{N(\bm{s})=j\}=\binom{m}{j}p^j(1-p)^{m-j}.
\end{equation}

\textbf{Case 2.} The free-ride codeword $\bm{s}$ is not the transmitted one, i.e., $\bm{s}\neq\bm{w}$.
In this case, we have
\begin{equation}
    \widetilde{\bm{c}}=\bm{\widehat{y}}+\bm{s}=\bm{u}\bm{G}_0+(\bm{s}+\bm{w})+\bm{\widehat{z}}=\bm{u}\bm{G}_0+\bm{\widetilde{v}}\bm{G}_1+\bm{\widehat{z}},
\end{equation}
where $\widetilde{\bm{v}}\neq \bm{0}^{k_1}$ is a non-zero sequence corresponding to the free-ride codeword $\bm{s}+\bm{w}$.
Hence $\widetilde{\bm{c}}$ can be viewed as the received sequence by passing the LDPC codeword $\bm{c}$ through a BSC with crossover probability 1/2, as the components of the generator matrix $\bm{G}_1$ are randomly and independently generated.
Consequently, each parity check is not satisfied with a probability of 1/2, leading to an approximate distribution of $N(\bm{s})$ as
\begin{equation}\label{eqn:Ns}
    \mathrm{Pr}\{N(\bm{s})=j\}=\frac{1}{2^{m}}\binom{m}{j}.
\end{equation}

The difference between the distributions (\ref{eqn:Nw}) and (\ref{eqn:Ns}) are obvious and can be significant especially when $\rho$ is small and $p_b =Q(1/\sigma)\ll 1/2$.
In particular, the distribution of $N(\bm{s})$ with the transmitted $\bm{s}$ depends on the SNR, while the one with erroneous $\bm{s}$ is irrelevant to the SNR.

To verify these analysis, we provide the following example.
\begin{example}\label{exp:histo}
    We consider a rate-1/2~(3,6)-regular LDPC code of length 8064 for the payload data transmission as an example.
    We set $k_1=5$.
    The histograms of $N(\bm{s})$ with correct candidate and erroneous candidates are show in Fig.~\ref{fig:histo}, from which we see that $N(\bm{s})$ is likely to be large if $\bm{s}$ is erroneous.
    We also see that the simulations match well with the approximate distributions.
    As expected, the distribution of $N(\bm{w})$ shifts to left as SNR increases while the distribution of $N(\bm{s})$ with $\bm{s}\neq \bm{w}$ remains unchanged.
    As a result, the difference between the two distributions becomes significant as SNR increases.
\end{example}
\begin{figure}
    \centering
    \includegraphics[width=\figwidth]{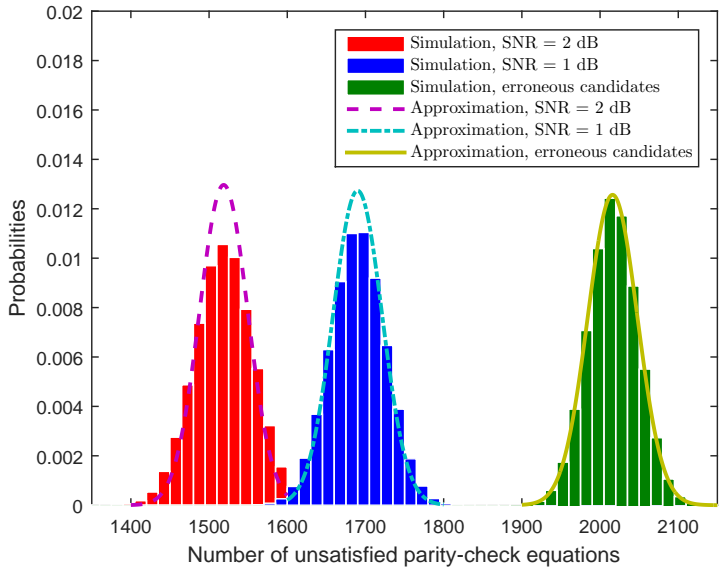}
    \caption{Histograms of the number of unsatisfied parity checks for the (3,6)-regular LDPC code $\mathscr{C}[8064, 4032]$.}
    \label{fig:histo}
\end{figure}
\subsection{Hard-decision Decoding}
The above discussion on the syndrome statistics shows a way to transmit the extra bits along with the LDPC coded payload data.
Let $\bm{G}_1$ be a  $k_1\times n$ randomly generated matrix  with each element drawn independently from a Bernoulli distribution with success probability $1/2$.
The extra bits are first encoded by the linear block code $\mathscr{C}_1$ with the generator matrix $\bm{G}_1$, and then superimposed on the LDPC coded payload data.
In principle, these extra bits can be decoded by an exhaustive search algorithm~\cite{Cai2019}, which outputs $\widehat{\bm{v}}$ such that $N(\widehat{\bm{v}}\bm{G}_1) = \min_{\bm{s}\in\mathscr{C}_1} N(\bm{s})$.
For completeness, we reproduce the HDD algorithm of the scheme with randomly coded extra bits in Algorithm~\ref{alg:decran}.

\begin{algorithm2e}\label{alg:decran}
    \caption{Hard-decision decoding of the scheme with randomly coded extra bits}
    1.~Make hard-decision on the received sequence $\bm{y}$, resulting in $\bm{\widehat{y}}$.

    2.~\For {all $\bm{s}\in \mathscr{C}_1$} {Compute $N(\bm{s})=W_H((\bm{\widehat{y}}+\bm{s})\bm{H}^{\mathrm T})$.}

    3.~Output $\widehat{\bm{v}}$ that corresponds to the lightest $N(\bm{s})$. That is, output $\widehat{\bm{v}}$ such that
        $$N(\widehat{\bm{v}}\bm{G}_1)=\min_{\bm{s}\in \mathscr{C}_1}N(\bm{s}).$$

    4.~Remove the interference of $\widehat{\bm{w}}=\widehat{\bm{v}}\bm{G}_1$ on $\bm{y}$, obtaining
        $$\widetilde{y}_i=(-1)^{\widehat{w}_i}y_i~~\textrm{for}~i=0,1,\cdots,n-1.$$

    5.~Input $\widetilde{\bm{y}}$ into the LDPC decoder to obtain the estimated payload data $\widehat{\bm{u}}$.
\end{algorithm2e}
\begin{remark}
In Step~2 of Algorithm~\ref{alg:decran}, the computational loads can be reduced at the expense of the memory loads.
Indeed, the weight $N(\bm{s})$ can be computed by $N(\bm{s})=W_H(\bm{\widehat{y}}\bm{H}^{\mathrm T}+\bm{s}\bm{H}^{\mathrm T})$, where $\bm{s}\bm{H}^{\mathrm T}$ can be pre-computed off-line and stored for use since they do not depend on the channel observations.
\end{remark}
\subsection{Performance and Complexity Analysis}
By HDD, the WER of the extra bits is given by
\begin{align}
    \notag \lambda_1 = & \mathrm{Pr}\{\textrm{the~decoder~outputs}~\bm{s}\neq\bm{w}\} \\
    \leqslant & \mathrm{Pr}\{\min_{\bm{s}\neq\bm{w}} N(\bm{s})\leqslant N(\bm{w})\},
\end{align}
where the distributions of $N(\bm{w})$ and $N(\bm{s})$ can be approximated by~(\ref{eqn:Nw}) and (\ref{eqn:Ns}), respectively.

According to the central-limit theorem, for a large $n$, $N(\bm{w})$ can be further approximated by a normal distribution with mean $\mu_0=\frac{m}{2}[1-(1-2p_b)^{\rho}]$ and variance $\sigma_0^2=\frac{m}{4}[1-(1-2p_b)^{2\rho}]$, and $N(\bm{s})$ can be approximated by a normal distribution with mean $\mu_1=\frac{m}{2}$ and variance $\sigma_1^2=\frac{m}{4}$.
By these approximations, $\lambda_1$ can be calculated by
\begin{equation}\label{eqn:lam1app}
    \lambda_1 \approx 
    1-\frac{1}{\sqrt{2\pi}\sigma_0}\int_{\mathbb{R}}\left[Q\left(\frac{2t-m}{\sqrt{m}}\right)\right]^{2^{k_1}-1}\exp\left(-\frac{(t-\mu_0)^2}{2\sigma_0^2}\right)\mathrm{d}t
\end{equation}

Compared to the conventional payload data transmission, extra computational loads are required for decoding the extra transmission.
The number of extra operations is $\mathcal{O}(2^{k_1}n)$, which grows exponentially with the number $k_1$ of extra bits.
Therefore, due to the limitation of computational resources, only a very few extra bits can be transmitted.

\subsection{Soft-decision Decoding}
Like payload codes, free-ride codes can also be decoded by employing soft decisions to improve the performance.
To this end, we turn to the log-likelihood ratio~(LLR), defined as
\begin{equation}\label{eqn:llrx}
    \Lambda_j(\bm{x})=\log\frac{P_{Y|X}(y_j|0)}{P_{Y|X}(y_j|1)}~~\textrm{for}~j=0,1,\dots,n-1.
\end{equation}
The basic idea is the same as that for HDD.
We check each $\bm{s}\in \mathscr{C}_1$ to find the right one.

For each $\bm{s}\in \mathscr{C}_1$, the LLRs with respect to $\bm{x}+\bm{s}$ can be obtained by flipping,
\begin{equation}\label{eqn:tildeL}
    \Lambda_j(\bm{x}+\bm{s})=(-1)^{s_j}\Lambda_j(\bm{x})~~\textrm{for}~j=0,1,\dots,n-1.
\end{equation}
Then the $i$-th parity check is satisfied with LLR given by
\begin{align}\label{eqn:llr}
    \notag \Lambda_i((\bm{x}+\bm{s})\bm{H}^{\mathrm{T}})=& \log\frac{\mathrm{Pr}\{\textrm{the}~i\textrm{-th~parity~check~is~satisfied}\}}{\mathrm{Pr}\{\textrm{the}~i\textrm{-th~parity~check~is~not~satisfied}\}}\\
    =&2\tanh^{-1}\left(\prod_{j:h_{ij}=1}\tanh\left(\frac{1}{2}\Lambda_j(\bm{c})\right)\right)~~\textrm{for}~i=0,1,\dots,m-1.
\end{align}
For each $\bm{s}\in \mathscr{C}_1$, we define
\begin{equation}\label{eqn:Lambda}
    \Lambda(\bm{s})=\sum_{i=0}^{m-1} \Lambda_i((\bm{x}+\bm{s})\bm{H}^{\mathrm{T}}).
\end{equation}
%

Intuitively, we can imagine that  $\bm{s}=\bm{w}$ will lead to a larger $\Lambda(\bm{s})$ since, in this case, $\bm{x}+\bm{s}=\bm{c}\in\mathscr{C}_0$ and the parity-check equations should hold more likely especially for high SNRs.
This can also be verified by statistical learning which is omitted here.
To summarize, we present the SDD algorithm in Algorithm~\ref{alg:decran_soft}.
\begin{algorithm2e}\label{alg:decran_soft}
    \caption{Soft-decision decoding of the scheme with randomly coded extra bits}
    1.~Compute the LLR $\bm{\Lambda}(\bm{x})$ corresponding to the received sequence $\bm{y}$.

    2.~\For{ all $\bm{s}\in \mathscr{C}_1$}{
    Compute $\Lambda(\bm{s})$ according to~(\ref{eqn:tildeL})-(\ref{eqn:Lambda}).}

    3.~Output $\widehat{\bm{v}}$ that corresponds to the largest $\Lambda(\bm{s})$. That is, output $\widehat{\bm{v}}$ such that
        $$\Lambda(\widehat{\bm{v}}\bm{G}_1)=\max_{\bm{s}\in \mathscr{C}_1}\Lambda(\bm{s}).$$

    4.~Remove the interference of $\widehat{\bm{w}}=\widehat{\bm{v}}\bm{G}_1$ on $\bm{y}$, obtaining
        $$\widetilde{y}_i=(-1)^{\widehat{w}_i}y_i~~\textrm{for}~i=0,1,\cdots,n-1.$$

    5.~Input $\widetilde{\bm{y}}$ into the LDPC decoder to obtain the estimated payload data $\widehat{\bm{u}}$.
\end{algorithm2e}

\begin{example}\label{exp:ranWER}
    The rate-1/2 (3,6)-regular LDPC code $\mathscr{C}[8064, 4032]$ in \textbf{Example~\ref{exp:histo}} is taken as the payload code.
    As shown in Fig.~\ref{fig:werran}, the WER performances of hard-decision decoding of the extra bits are well predicted by~(\ref{eqn:lam1app}).
    For the proposed scheme with $k_1 =5$, we observe that, at SNR 0.5~dB, the scheme achieves a WER lower than $10^{-6}$, which can be traded off with the coding rate for the extra transmission.
    We see that the WER performance is further improved while using SDD, indicating the potentials to carry more extra bits.
    Then the sum-product algorithm with maximum 50 iterations is employed for decoding the LDPC coded payload data.
    We present in Fig.~\ref{fig:berran} the BER of the code $\mathscr{C}[8064, 4032]$.
    It can be seen that the transmission of extra bits has noticeable effects on the payload data in the region of $\textrm{SNR} \leqslant -0.5$~dB but negligible effects in the region of $\textrm{SNR} \geqslant 0$~dB, which is of interest for the coded payload data.
    This also suggests that the extra transmission is more reliable than the payload transmission, a figure of merit for some applications.
\end{example}
\begin{figure}
    \centering
    \includegraphics[width=\figwidth]{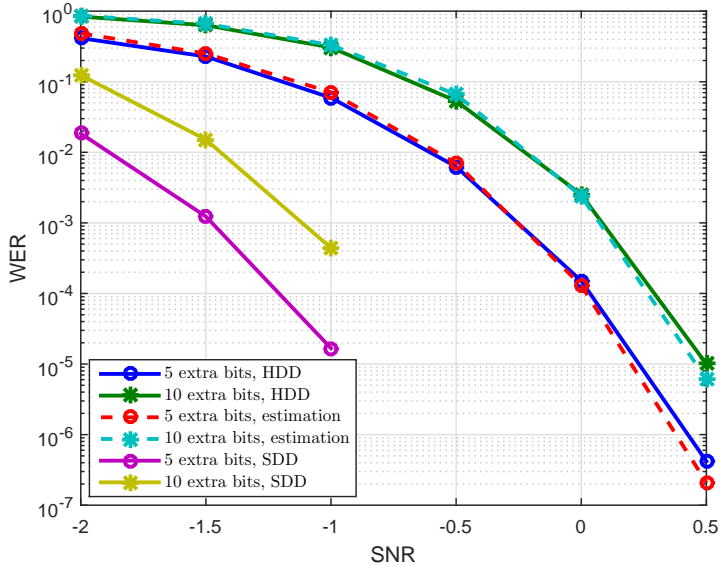}
    \caption{WER performance of the randomly coded extra bits superimposed on the (3,6)-regular LDPC code $\mathscr{C}[8064, 4032]$.}
    \label{fig:werran}
\end{figure}
\begin{figure}
    \centering
    \includegraphics[width=\figwidth]{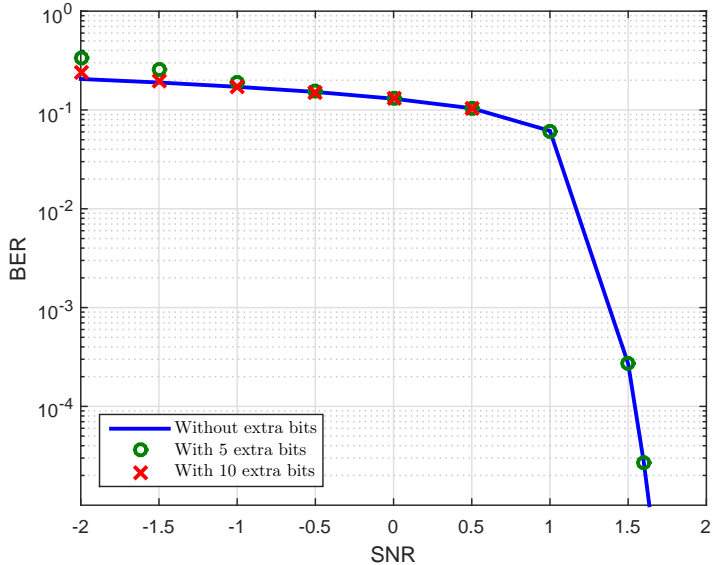}
    \caption{BER performance of (3,6)-regular LDPC $\mathscr{C}[8064, 4032]$ coded payload data, where extra bits are encoded by random superposition and decoded with hard decision.}
    \label{fig:berran}
\end{figure}
\section{Structured Free-ride Codes}
The main issue of the random free-ride codes is the decoding complexity, which grows exponentially with the number of extra bits.
In this section, we develop a general construction for free-ride codes and then presented as examples the repetition codes and the RM codes.

\subsection{Syndrome Channel}
As illustrated in Fig.~\ref{fig:model} and characterized by~(\ref{eqn:blockchannel}), the extra transmission link $\bm{w}\rightarrow\bm{y}$ can be viewed as a block-wise memoryless channel, over which the free-ride codeword is corrupted by two types of ``noises''.
One is the interference from the payload codeword, while the other is the channel noise.
On one hand, to recover the payload data by successive cancellation decoding, we need first recover the extra bits.
On the other hand, to recover the extra bits, we need null the interference from the payload link.
This can be achieved by transforming the hard-decision vector into a syndrome with the parity-check matrix $\bm{H}$ of the payload code.

More formally, upon receiving $\bm{y}$, we make a hard decision, obtaining $\bm{\widehat{y}}$.
The channel can be modeled as
\begin{equation}\label{eqn:hard}
    \bm{\widehat{y}} = \bm{c} + \bm{w} + \bm{\widehat{z}},
\end{equation}
where $\bm{c}$, $\bm{w}$ and $\bm{\widehat{z}}$ are the transmitted payload codeword, the free-ride codeword and the error pattern, respectively, all of which are binary vectors.
For HDD, the channel is a BSC with cross probability
\begin{equation}
    \mathrm{Pr}\{\widehat{Z}_j=1\}=p_b=Q(\frac{1}{\sigma})~~\textrm{for}~j=0,1,\dots,n-1.
\end{equation}
For SDD, the channel is a time-varying BSC with cross probabilities
\begin{equation}\label{eqn:soft}
    \mathrm{Pr}\{\widehat{Z}_j=1\}=\frac{\min\{P_{Y|X}(y_j|0),P_{Y|X}(y_j|1)\}}{P_{Y|X}(y_j|0)+P_{Y|X}(y_j|1)}~~\textrm{for}~j=0,1,\dots,n-1.
\end{equation}
By multiplying the parity-check matrix on both sides of~(\ref{eqn:hard}), we have the {\em syndrome channel},
\begin{equation}
\bm{\widehat{y}}\bm{H}^{\mathrm{T}} = \bm{w}\bm{H}^{\mathrm{T}} + \bm{\widehat{z}}\bm{H}^{\mathrm{T}},
\end{equation}
where $\bm{w}\bm{H}^{\mathrm{T}}$ with $\bm{w} \in \mathscr{C}_1$ is the input and $\bm{\widehat{z}}\bm{H}^{\mathrm{T}}$ is the error pattern, both of which are binary vectors in $\mathbb{F}_2^m$.
This transformation is similar to the zero-forcing equalization commonly used in multiple-input multiple-output~(MIMO) detection, which cancels the interference from the payload link but usually introduces memory among the components of $\bm{\widehat{z}}\bm{H}^{\mathrm{T}}$.
The syndrome channel can be {\em approximated} as a memoryless BSC, either time-invariant or time-varying, as specified in the following.
Denote by $(\widehat{\bm{z}}\bm{H}^{\mathrm{T}})_i$ the $i$-th component of $\widehat{\bm{z}}\bm{H}^{\mathrm{T}}$.
For HDD, the cross probability is calculated as $\mathrm{Pr}\{(\widehat{\bm{Z}}\bm{H}^{\mathrm{T}})_i=1\}=p$ as expressed in~(\ref{eqn:calpc}),
while for SDD, the cross probabilities can be computed from~(\ref{eqn:soft}) by 
\begin{equation}
    \mathrm{Pr}\{(\widehat{\bm{Z}}\bm{H}^{\mathrm{T}})_i=1\}=\frac{1}{2}\left[1-\prod_{j:h_{ij}=1}(\mathrm{Pr}\{\widehat{Z}_j=0\}-\mathrm{Pr}\{\widehat{Z}_j=1\})\right].
\end{equation}

Since the payload codeword has been nulled in the syndrome channel, the extra bits can be recovered from $\bm{\widehat{y}}\bm{H}^{\mathrm{T}}$.
Hence, the performance is determined by the code $\mathscr{C}_1\bm{H}^{\mathrm{T}}\in \mathbb{F}_2^m$.
Actually, the basic idea is to transmit the extra bits by employing the coset codes.
To be precise, the payload code $\mathscr{C}_0$ define a coset partition $\mathbb{F}_2^{n}/\mathscr{C}_0$.
Then the construction of free-ride code $\mathscr{C}_1$ can be viewed as selecting some representative elements from the cosets.
For the free-ride code to be good, its codewords should be selected from distinct cosets.
The whole code $\mathscr{C}$ can be viewed as the direct sum of $\mathscr{C}_0$ and $\mathscr{C}_1$,
\begin{equation}
    \mathscr{C}=\bigcup_{\bm{w}\in\mathscr{C}_1}\left(\bm{w}+\mathscr{C}_0\right)=\mathscr{C}_1+\mathscr{C}_0
\end{equation}
By considering the linear mapping defined by the parity-matrix $\bm{H}$ from $\mathbb{F}_2^n$ to $\mathbb{F}_2^{m}$, we have the following diagram
\begin{equation}
\begin{array}{ccc}
\mathbb{F}_2^n&\overset{\bm{H}^{\mathrm{T}}}\longrightarrow&\mathbb{F}_2^{m}\\
\mathrel{\rotatebox{90}{$\subseteq$}}& &\mathrel{\rotatebox{90}{$\subseteq$}}\\
\mathscr{C}_0+\mathscr{C}_1 &\overset{\bm{H}^{\mathrm{T}}}\longrightarrow& \mathscr{C}_1\bm{H}^{\mathrm{T}}
\end{array}
\end{equation}
where $\mathscr{C}_1\bm{H}^{\mathrm{T}}=\{\bm{s}\bm{H}^{\mathrm{T}}:\bm{s}\in \mathscr{C}_1\}$ is referred to as the \emph{syndrome code} and denoted by $\mathscr{C}_s$.


Without loss of generality, we assume that the first $m$ columns of $\bm{H}$ are linearly independent.
That is, the parity-check matrix can be written as $\bm{H} = [\bm{H}_1, \bm{H}_2]$ with $\bm{H}_1$ being invertible.
Then we have the following general procedure to construct a free-ride code.
\begin{enumerate}
\item Construct a syndrome code $\mathscr{C}_s[m, k_1]$.
    Denote its generator matrix by $\bm{G}_s$.
\item Construct a free-ride code $\mathscr{C}_1[n, k_1]$ by specifying a generator matrix $\bm{G}_1$ such that
    \begin{equation}
        \bm{G}_1\bm{H}^{\mathrm T}=\bm{G}_s.
    \end{equation}
    For example, we can choose
    \begin{equation}
        \bm{G}_1 = [\bm{G}_s(\bm{H}_1^{-1})^{\mathrm T},\bm{0}_{k_1\times k}].
    \end{equation}
\end{enumerate}


The encoding of the free-ride code can be implemented directly over $\mathbb{F}_2^n$ by calculating $\bm{w}=\bm{v}\bm{G}_1$ or indirectly over $\mathbb{F}_2^m$ by first calculating $\bm{w}_s=\bm{v}\bm{G}_s$ and then finding a pre-image of $\bm{w}_s$ in $\mathbb{F}_2^n$ with respect to the linear mapping defined by $\bm{H}$.
For example, we can choose
\begin{equation}
    \bm{w}=(\bm{w}_s(\bm{H}^{-1})^{\mathrm{T}},\bm{0}^k).
\end{equation}
The decoding of the free-ride code can be reduced to the decoding of the syndrome code over the syndrome channel.
In particular, for HDD, the ML criterion is equivalent to the minimum distance criterion 
\begin{equation}
    \bm{v}_{\textrm{ML}}=\underset{\bm{v}\in \mathbb{F}_2^{k_1}}{\arg\min}\{W_H(\bm{\widehat{y}}\bm{H}^{\mathrm{T}}+\bm{v}\bm{G}_s)\}.
\end{equation}
For SDD, the LLR corresponding to the syndrome channel is given by
\begin{equation}
    \Lambda_i(\bm{w}_s)=2\tanh^{-1}\left(\prod_{j:h_{ij}=1}\tanh\left(\frac{1}{2}\Lambda_j(\bm{x})\right)\right)~~\textrm{for}~i=0,1,\dots,m-1,
\end{equation}
where $\bm{\Lambda}(\bm{x})$ is obtained from the channel observations, defined by~(\ref{eqn:llrx}).
Then the extra bits can be estimated according to the ML criterion
\begin{equation}
    \bm{v}_{\textrm{ML}}=\underset{\bm{v}\in \mathbb{F}_2^{k_1}}{\arg\max}\langle(-1)^{\bm{v}\bm{G}_s},\bm{\Lambda}(\bm{w}_s)\rangle,
\end{equation}
where $\langle\bm{x},\bm{y}\rangle=\sum_i x_iy_i$ is the inner product.
\subsection{Repetition Codes as Syndrome Codes}
\subsubsection{Transmission of One Extra Bit}
To begin with, we consider the case that only one extra bit is transmitted, i.e., $v\in\mathbb{F}_2$.
Simply, we implement the repetition code as the syndrome code, i.e., $\bm{G}_s=(\bm{1}^m)$, the all-one vector with length $m$.
Then the extra bit can be decoded by a majority-logic~(MLG) decoder.
It is expected that the repetition coded extra bit can be better detected than the randomly coded one, since the repetition code has a larger minimum Hamming distance.

\begin{example}
    We consider the transmission of one extra bit, in which a rate-1/2 (3,6)-regular LDPC code with length 128 is used as the payload code.
    The WER performance of the extra bit and the payload data are shown in Fig.~\ref{fig:werrep1} and Fig.~\ref{fig:werbep1}, respectively.
    As predicted, the repetition coded extra bit achieves a lower error rate than the randomly coded one.
    We also see that, with SDD, the repetition coded extra bit has a negligible effect on the reliability of the payload data.
\end{example}
\begin{figure}
    \centering
    \includegraphics[width=\figwidth]{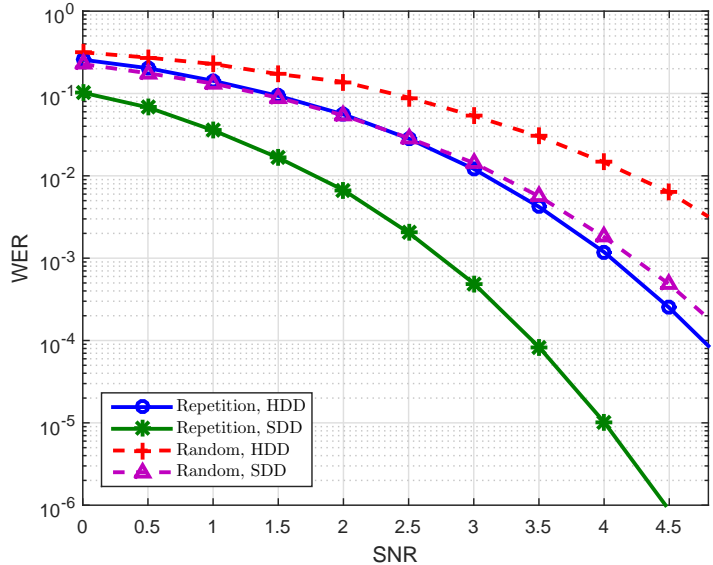}
    \caption{WER performance of repetition coded extra bit superimposed on the (3,6)-regular LDPC code $\mathscr{C}[128, 64]$.}
    \label{fig:werrep1}
\end{figure}

\begin{figure}
    \centering
    \includegraphics[width=\figwidth]{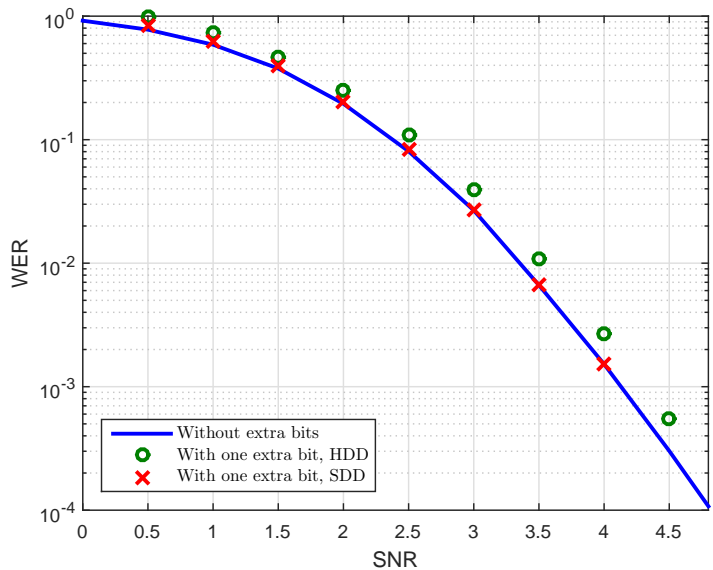}
    \caption{WER performance of the (3,6)-regular LDPC code $\mathscr{C}[128, 64]$ with one repetition coded extra bit.}
    \label{fig:werbep1}
\end{figure}
\subsubsection{Transmission of Multiple Extra Bits}
Assume that $k_1$ extra bits are required to be transmitted along with an LDPC codeword and $m_1=m/k_1$ is an integer.
We use the $k_1$-fold Cartesian product of the repetition code as the syndrome code.
The generator matrix of $\mathscr{C}_s$ is a $k_1\times m$ matrix, given by
\begin{equation}
    \bm{G}_s=\mathrm{diag}\{\bm{1}^{m_1},\dots,\bm{1}^{m_1}\}.
\end{equation}
By employing the MLG decoding of the syndrome code, the number of extra operations is $\mathcal{O}(n)$, which is irrelevant to the number of extra bits $k_1$.


\begin{example}\label{exp:repWER2}
    The rate-1/2 (3,6)-regular LDPC code $\mathscr{C}[8064, 4032]$ in \textbf{Example~\ref{exp:histo}} is used for the payload transmission.
    The WER performance of the repetition coded extra bits is shown in Fig.~\ref{fig:werrep2}, where the dashed curves correspond to performance estimates
    \begin{equation}
        {\lambda}_1\approx \sum_{i=\lceil m_1/2\rceil}^{m_1}\binom{m_1}{i}p^i(1-p)^{m_1-i},
    \end{equation}
    obtained by assuming the syndrome channel as an ideal BSC($p$) with  $p$ given in~(\ref{eqn:calpc}).
    We see that the simulation results with HDD match well with the estimates.
    We also see that, with SDD, 18 extra bits can be transmitted with a WER about $10^{-4}$ at $\textrm{SNR}=1.5$~dB by the use of repetition codes.
\end{example}
\begin{figure}
    \centering
    \includegraphics[width=\figwidth]{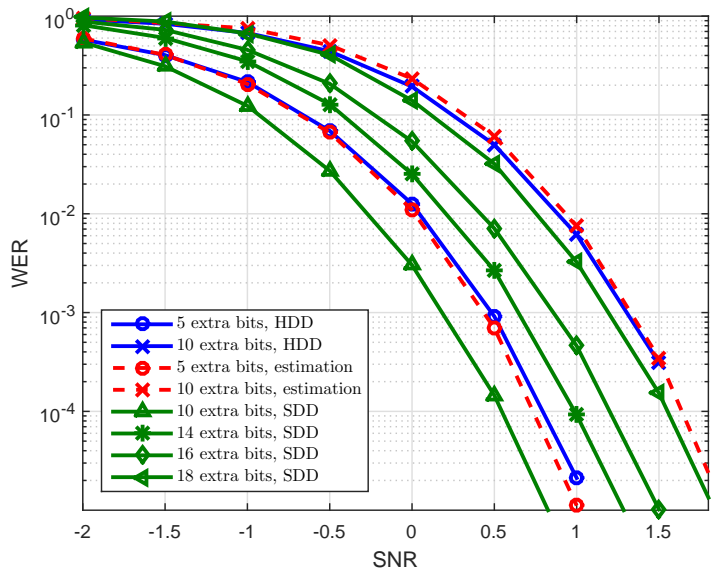}
    \caption{WER performance of the repetition coded extra bits superimposed on the (3,6)-regular LDPC code $\mathscr{C}[8064, 4032]$.}
    \label{fig:werrep2}
\end{figure}

\subsection{Reed-Muller Codes as Syndrome Codes}
By comparing the results in \textbf{Example~\ref{exp:ranWER}} and \textbf{Example~\ref{exp:repWER2}}, we see that  the repetition coded extra bits have a worse error performance than the randomly coded ones.
The degradation of the performance is due to the small minimum distance of the Cartesian product of repetition codes.
To tackle this issue, we consider the first-order RM code as the syndrome code, which has both a large minimum distance and a low complexity ML decoding algorithm.
We use $\mathrm{RM}(1,\eta)$ to denote the first-order RM code with dimension $\eta+1$ and length $2^\eta$.
Since the minimum distance of $\mathrm{RM}(1,\eta)$ is half of its length, it is expected that the RM coded extra bits can have comparable error performance to the randomly coded ones.
In order to pack more extra bits into an LDPC codeword, Cartesian products of RM codes are implemented as syndrome codes.
Note that to fit the length of payload code, punctured $\mathrm{RM}(1,\eta)$ codes can be implemented in practice.
The ML decoding of the first order $\mathrm{RM}(1,\eta)$ can be implemented by the Fast Fourier Transform~(FFT) algorithm~\cite{MacWilliams1977}.
As a result, the number of extra operations can be reduced to $\mathcal{O}(n\log n)$, which is also irrelevant to the number of extra bits $k_1$.

\begin{example}
    The rate-1/2 (3,6)-regular LDPC code $\mathscr{C}[8064, 4032]$ in \textbf{Example~\ref{exp:histo}} is used for the payload transmission.
    FFT based soft-decision ML decoding algorithm is employed for decoding the extra bits.
    The WER performance of the RM coded extra bits is shown in Fig.~\ref{fig:werrm}.
    We see that, when $k_1=10$, the RM coded extra bits, with a much lower decoding complexity, can achieve a slightly better performance than the randomly coded ones.
    The BER of the payload data is shown in Fig.~\ref{fig:berrm}.
    We see that as many as 60 extra bits can be transmitted with a negligible effect on the reliability of the payload data.
\end{example}
\begin{figure}
    \centering
    \includegraphics[width=\figwidth]{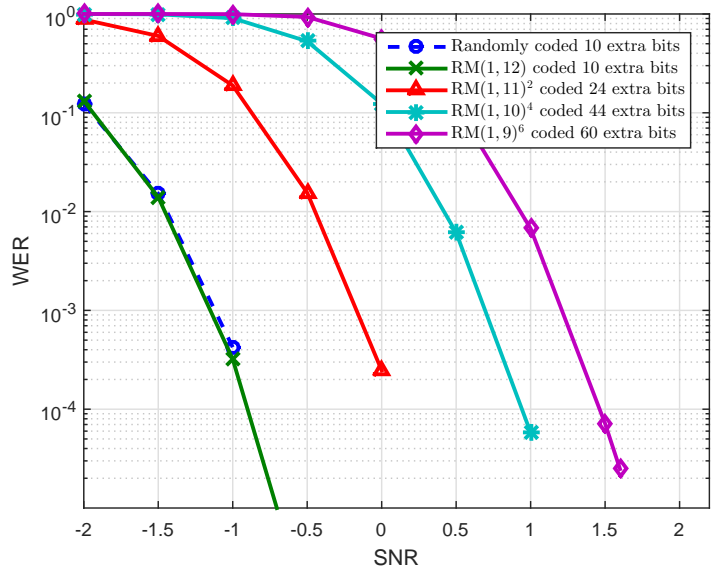}
    \caption{WER performance of RM coded extra bits superimposed on the (3,6)-regular LDPC code $\mathscr{C}[8064, 4032]$.}
    \label{fig:werrm}
\end{figure}
\begin{figure}
    \centering
    \includegraphics[width=\figwidth]{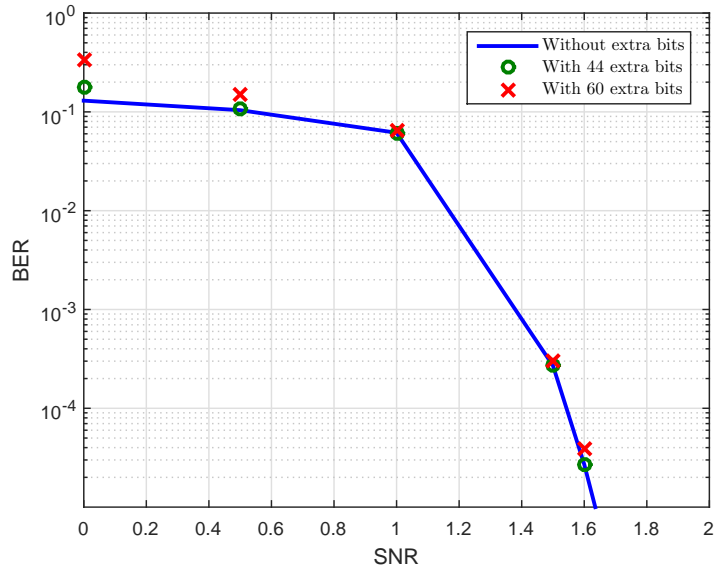}
    \caption{BER performance of the (3,6)-regular LDPC code $\mathscr{C}[8064, 4032]$ with RM coded extra bits.}
    \label{fig:berrm}
\end{figure}

\section{Conclusions}
In this paper, we investigated the problem of transmitting extra bits by superposition over the coded payload link without any cost of extra transmission energy or extra bandwidth.
We derived the accessible capacity of extra transmission and provided a lower bound on the accessible capacity.
For LDPC coded payload links, we first presented the random free-ride coding scheme, which however requires exponential extra computational loads for decoding.
To lower down the decoding complexity, we then presented structured free-ride coding schemes.
Numerical results showed that, based on RM codes, as many as 60 extra bits can be transmitted reliably over a payload link coded with a length-8064 LDPC code.

\bibliographystyle{IEEEtran}
\bibliography{freeride_cai}
%

%

%
%
%




\end{document}